\documentclass[a4paper]{article}
\usepackage[english]{babel}
\usepackage[utf8]{inputenc}
\usepackage{amsmath, amssymb, amsthm}
\usepackage{graphics}
\usepackage{todonotes}
\usepackage{physics}
\usepackage[margin=0.9in]{geometry}
\usepackage{hyperref}
\usepackage{graphicx}
\usepackage{physics}
\usepackage{mathrsfs}
\usepackage{verbatim}

\newcommand{\cc}{\leftrightarrow}

\usepackage{tikz}
\usetikzlibrary{positioning}
\usetikzlibrary{arrows}

\newtheorem{theorem}{Theorem}

\newtheorem{corollary}[theorem]{Corollary}
\newtheorem{conjecture}[theorem]{Conjecture}

\newtheorem{question}[theorem]{Question}

\newlength\tindent
\title{On monotonicity and couplings of random currents  and the loop-$\mathrm{O}(1)$-model}
\author{Frederik Ravn Klausen \\
QMATH, University of Copenhagen\\
klausen@math.ku.dk}
\date{\today}

\begin{document}

\maketitle

\begin{abstract}
\noindent Using recent couplings, we provide counterexamples to monotonicity properties of percolation models related to graphical representations of the Ising model. We further prove a new coupling of the double random current model to the loop-$\mathrm{O}(1)$-model. 
\end{abstract}

\section{Introduction} 
The field of graphical representations of the Ising model has been developing rapidly in the last 50 years. Examples of well-studied models include the random cluster model \cite{FK70}, \cite{Gri06}, the random current representation \cite{GHS70}, \cite{ABF87}, \cite{Aiz82}, \cite{AD19} as well as the loop-$\mathrm{O}(1)$-model \cite{KW41}. The motivation for studying these models is that one can link percolative and connectivity properties to phase transitions and correlation functions of the Ising model.
Usually the models are introduced in their own right, but for the purpose here we define the models directly as percolation models via their couplings to the loop-$\mathrm{O}(1)$-model \cite[Exercise 36]{DC17}, \cite[Theorem 3.5]{GJ09}, \cite{LW16} \cite{Lis}. The framework provided by these recent couplings greatly simplifies the proofs and counterexamples given here.

For the definitions, we need the notion of an even subgraph. 
Given a finite graph $G = (V,E)$ an \emph{ even subgraph of } (V,F)  is a spanning subgraph where each $ v \in V $  is incident to an even number of edges in $F$.   The set of even subgraphs of a graph $G$ is denoted $ \mathcal{E}_\emptyset(G) $. A natural probability measure on $ \mathcal{E}_\emptyset(G) $ is the \emph{loop-$\mathrm{O}(1)$-model} $l_x$  which is given by  
\begin{align}
l_x(g) = \frac{1}{Z} x^{\abs{g}}, \text{     for each } g  \in \mathcal{E}_\emptyset(G)
\end{align}
with $Z = \sum_{g \in \mathcal{E}_\emptyset(G)} x^{\abs{g}} $. 
 Here $\abs{g}$ denotes the number of edges in $g$ and $x = \tanh(\beta)$ as in \cite{Lis}. 
 
 For two configurations $\omega_1, \omega_2$ of edges in the graph $G$ we let $\omega_1 \cup \omega_2$ be the configuration with the union of the edges in $\omega_1 $ and $\omega_2$ and for two probability measures $\mu_1,\mu_2$ on the configurations of edges let $\mu_1 \cup \mu_2 $ be the measure corresponding to sampling independent configurations with the law of $\mu_1, \mu_2$ and taking the union of the two configurations. We define $\mu^{\otimes 2} = \mu \cup \mu $ and let $\mathbb{P}_x$ be Bernoulli edge percolation with parameter $x \in (0,1)$. Then we can define the (traced, sourceless) \textit{single random current}  at inverse temperature $\beta$  as
\begin{align}
P_x = l_x \cup \mathbb{P}_{1- \sqrt{1-x^2}}
\end{align}
 and the (FK-Ising, $q=2$) \textit{random cluster model} with $p = 1 - \exp(- 2\beta)$ by 
 \begin{align}
 \phi_\beta  = \phi_x = l_x \cup \mathbb{P}_x.
 \end{align}
In Appendix A3 we describe how the previous definitions are related to the standard definitions in the literature.
Using that $\mathbb{P}_p \cup  \mathbb{P}_q =  \mathbb{P}_{p+q - pq} $ we obtain as in \ (traced, sourceless) \emph{double random current} which is the union of two independent single random currents is given by 
\begin{align} P_x^{\otimes 2} = l_x^{\otimes 2} \cup \mathbb{P}_{x^2}
\end{align}
 and similarly the double random cluster model is $\phi_x^{\otimes 2} = l_x^{\otimes 2} \cup \mathbb{P}_{x(2-x)}$.

These models are motivated by being graphical representations of the Ising model. In particular, we have the following relation between correlation functions \cite[(1.5), (4.6)]{DC17} 
\begin{align}
P_x^{\otimes 2}(x \cc y) = \langle \sigma_x \sigma_y \rangle^2 = \phi_{ \beta }( x \cc y )^2  \text{      } x,y \in V. 
\end{align}
where $\sigma_x , \sigma_y$ are Ising spins and $\langle \cdot \rangle $ is the expectation of the Ising measure. 

Many tools are available for the random cluster model in particular monotonicity and the FKG-inequality {\cite[Theorem 1.6]{DC17}}.  In the following, we investigate the monotonicity properties of the other models some of which turn out to be less well behaved than the random cluster model. 

The first question is motivated by the fact that, by monotonicity of Ising correlations or of the random cluster model, it follows that $ x \mapsto P_x^{\otimes 2}(x \cc y) = \langle \sigma_x \sigma_y \rangle^2 = \phi_\beta( x \cc y )^2   $ is increasing. Here and in the following, we define $\{ A \cc B \}$ to be the event that at least one vertex in $A$ is connected to a vertex in $B$.
\begin{question}\cite[Question 2]{DC16} \label{dcq}
	Is the map
$
	x \mapsto P_x^{\otimes 2}(  A \cc B) 
$
	increasing for any subsets $A,B \subset V$. 
\end{question} 
\noindent In general, we are also interested in whether the model in itself is monotonic.  The next question is related and with Theorem \ref{sumthm} in mind it can be seen as an easier example of the question from before.

\begin{conjecture}\cite[Conjecture 5.1]{GMM18} \label{evencon}
	Let $G$ be an even graph then $l_x$ is monotonic. 
\end{conjecture}

\noindent In the following, we give a counterexample to Conjecture \ref{evencon} and we give partial results and counterexamples  towards Question \ref{dcq}. Most notably we show that the single random current measure is \emph{not} monotonic.  In general, for each of the above mentioned percolation models, we try to understand whether the FKG inequality, monotonicity, monotonicity of connection events $\{A \cc B \}$ and monotonicity of singleton connection events $\{ a \cc b \}$, hold. We denote these properties by FKG, MON, CON and SING respectively.

\begin{table}[ht]
	\caption{Overview of monotonicity properties. }
	\label{table:overview}
	\centering
	\begin{tabular}{c c c c c c c}
		\hline
		Case & $l_x$ & $P_x$ & $\phi_x$ &  $l_x^{\otimes 2}$ & $P_x^{\otimes 2}$  & $\phi_x^{\otimes 2}$  \\ [0.5ex] 
		$p(x)$ & 0  & $1- \sqrt{1-x^2}$ &$ x$ &  0  &$  x^2 $ &$ x(2-x)$  \\  
		\hline
		FKG & $\times $ & $\times $  & \checkmark & $\times $  & ?  &  \checkmark   \\
		MON &$\times $  & $\times $ &  \checkmark  &$\times $  & ?  &  \checkmark \\
		CON &$\times $ &$\times $  & \checkmark  & $\times $ & ? &  \checkmark  \\
		SING &$\times $ & $\times $ & \checkmark  &   $\times $  &  \checkmark  & \checkmark  \\
		\hline
	\end{tabular}
	\label{table:nonlin}
\end{table}

The double current is more well behaved as it satisfies (SING), we further prove a new coupling for the double current and discuss its monotonicity properties. Some of the information in the theorems and counterexamples given is summed up in Table \ref{table:overview}.

\section{Counterexamples}

In this section we give counterexamples to the FKG and monotonicity for the models $l_x, P_x,$ and $l_x^{\otimes 2}$ .
\\
\\
\textbf{2.1 Counterexamples to FKG.}
Consider the graph in Figure \ref{fkgcounter} with $n =l$.  
The partition function corresponding to this graph is 
\begin{align*}
Z =    \sum_{g \in \mathcal{E}_\emptyset(G)} x^{\abs{g}}  =1+ x^{n+m} + x^{n+m} + x^{2n}.
\end{align*}
 Let $X_1$ be the event that all edges in the upper $n+m$ loop are open and let $X_2$ be the event that all edges in the lower $n+m$ loop are open. 
\\
\\
\noindent \textbf{2.1.1 Loop-$\mathrm{O}(1)$-model: } Notice that $l_x(X_1 \cap X_2)  = 0$, whereas $l_x(X_1) > 0, l_x(X_2) > 0$ and there is no positive association.  A counterexample on an even graph is given in \cite{GMM18}.

\begin{figure}

	\begin{center}	{\caption{The graph used to construct counterexamples to FKG. The three segments consist of $n,m$ and $l$ edges respectively. In some examples we let $n=l$. \label{fkgcounter} }}{
		
		\includegraphics[scale = 0.2]{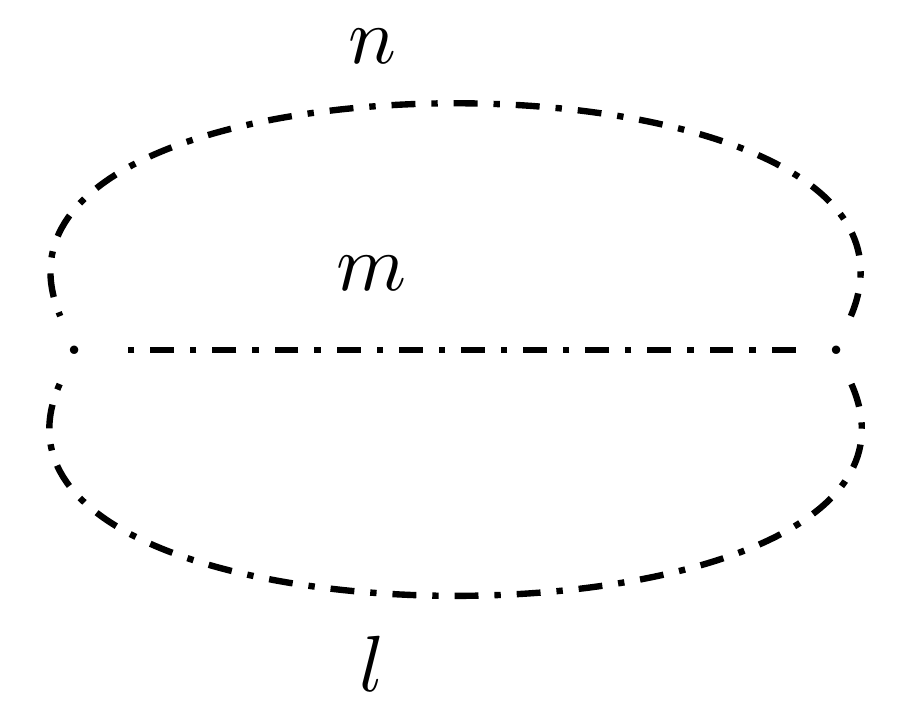}


		}
	\end{center}
\end{figure}
\vspace{0.2 cm}
\noindent \textbf{2.1.2 Single random currents:}
For the traced single random current we can use the same example and we get with $p(x)=1- \sqrt{1-x^2}$ that
\begin{align*}
Z P_x(X_2) = Z P_x(X_1) =  p^{n+m} + x^{n+m} + x^{n+m}p^n + x^{2n}p^m 
\end{align*} and 
\begin{align*} Z P_x(X_2 \cap X_1)  = p^{2n+m} + 2 \cdot x^{n+m}p^n + x^{2n}p^m. \end{align*}
Now, since $Z \to 1$ in the limit $x \to 0$ the function 
\begin{align*}
x \mapsto P_x(X_2 \cap X_1) -P_x(X_2)P_x(X_1) = \frac{Z P_x(X_2 \cap X_1) - Z P_x(X_2) Z P_x(X_1)}{Z^2} + \frac{(Z-1) P_x(X_2 \cap X_1)}{Z^2} 
\end{align*}
becomes negative for sufficiently small $x$. Thus, FKG is not satisfied.
Using the same example, but a slightly more complicated analysis the same counterexample works for $l_x^{\otimes 2}$  we give the details in the Appendix.

For the \textit{double random current} this example cannot be used to find a counterexample since indeed the two events $X_1$ and $X_2$ are positively associated for all $n$ and $m$. 
\\
\\
\textbf{2.2 Counterexamples to monotonicity.}
We now give a counterexample to Conjecture \ref{evencon}.  Consider the graph shown rightmost on Figure \ref{counter} with $n$ edges along the outer paths and $m$ edges along the inner paths for some even $m$. 
Let further $a$ and $b$ be the two points shown. Then there are 8 possible even subgraphs as shown on Figure \ref{counter} where we have also listed their corresponding weigths. We see that $Z = 1+x^{2n}+x^{2m}+4x^{n+m}+x^{2n+2m}$ and as $a$ and $b$ are connected in the 3rd and the last graph only and therefore
\begin{align*}
l_x(a \cc b) = \frac{x^{2m} + x^{2m+2n}}{Z}.
\end{align*}
Now, numerical inspection shows that this function is not monotonic for $m=2$ and $n \geq 8$. This provides a counterexample to Conjecture \ref{evencon}. The intuition behind the counterexample is that when $x$ increases we get an interval where it is more likely to sample one of the graphs with $n$ vertices where $a$ and $b$ are not connected. Similarly, we can use the coupling to the loop-$\mathrm{O}(1)$-model to calculate the probability that $a$ is connected to $b$ for the single random current and the double loop-$\mathrm{O}(1)$-model. We do the calculation in the appendix and show the plot of the functions in Figure \ref{plot} where we see that they are not monotone. Remark that $P_x^{\otimes 2}(a \cc b)$  is always monotone since it satisfies SING.

	\begin{figure}
		{\caption{The graph in question shown as the rightmost graph along with its eight even subgraphs.  We let the outer paths be $n$ edges long and the inner paths be $m$ edges long. The nodes $a$ and $b$ are marked with dots. We list number of edges of each subgraph, the corresponding weights and whether $a$ and $b$ are connected in the subgraph. \label{counter} }}
		
		\begin{center}
		{ 	\begin{tabular}{c |c | c | c | c | c | c | c| c  } 
				\hline
			  Subgraph &	\begin{tikzpicture}[scale = 0.5]
			\draw (1,1) node {.};
			\draw (1,-1) node {.};
			\draw (1,-2) ; 
			\end{tikzpicture}
			& 
			\begin{tikzpicture}[scale = 0.5]
			\draw  (0,0) -- (0,2) --(2,2)  -- (2,0)  (0,0) -- (0, - 2) --(2, - 2)  -- (2,0)  ;
			\draw (1,1) node {.};
			\draw (1,-1) node {.};
			\end{tikzpicture}
			&
			\begin{tikzpicture}[scale = 0.5]
			\draw (0,0) --(1,1) -- (2,0)  (0,0) -- (1,-1) -- (2,0)  ;
			\draw (1,1) node {.};
			\draw (1,-1) node {.};
			\draw (1,-2) ;
			\end{tikzpicture}
			&
			\begin{tikzpicture}[scale = 0.5]
			\draw   (0,0) -- (1,-1) -- (2,0)  (0,0) -- (0,2) --(2,2)  -- (2,0)    ;
			\draw (1,1) node {.};
			\draw (1,-1) node {.};
			\draw (1,-2) ;
			\end{tikzpicture}
			&
			\begin{tikzpicture}[scale = 0.5]
			\draw   (0,0) -- (1,-1) -- (2,0)   (0,0) -- (0, - 2) --(2, - 2)  -- (2,0)  ;
			\draw (1,1) node {.  };
			\draw (1,-1) node {. };
			\end{tikzpicture}
			&
			\begin{tikzpicture}[scale = 0.5]
			\draw (0,0) --(1,1) -- (2,0)   (0,0) -- (0, - 2) --(2, - 2)  -- (2,0)  ;
			\draw (1,1) node {.};
			\draw (1,-1) node {.};
			\end{tikzpicture}
			&
			\begin{tikzpicture}[scale = 0.5]
			\draw (0,0) --(1,1) -- (2,0)    (0,0) -- (0,2) --(2,2)  -- (2,0)   ;
			\draw (1,1) node {.};
			\draw (1,-1) node {.};
			\draw (1,-2) ;
			\end{tikzpicture}
			&
			\begin{tikzpicture}[scale = 0.5]
			\draw (0,0) --(1,1) -- (2,0)  (0,0) -- (1,-1) -- (2,0)  (0,0) -- (0,2) --(2,2)  -- (2,0)  (0,0) -- (0, - 2) --(2, - 2)  -- (2,0)  ;
			\draw (1,1) node {.};
			\draw (1,-1) node {.};
			\draw (1,1.2) node {\tiny{$m$}};
			\draw (1,-1.2) node {\tiny{$m$}};
			\draw (1,2.2) node {\tiny{$n$}};
			\draw (1,-1.8) node {\tiny{$n$}};
			\end{tikzpicture} \\
			\hline 
		 	Edges & 0 & $2n$ &  $2m$ & $n+m$ & $n+m$ & $n+m$ & $n+m$ &  $2m +2n$ \\
			\hline 
			Weight & 1 & $x^{2n}$ &  $x^{2m}$ & $x^{n+m}$ & $x^{n+m}$ & $x^{n+m}$ & $x^{n+m}$ &  $x^{2m +2n}$ \\
				\hline 
			$\{a \cc b \} $ & $\times$  & $\times $ &  \checkmark &  $\times $ &  $\times $ &  $\times $ &  $\times $ & \checkmark\\
				\hline 
			\end{tabular}}
		\end{center}
	\end{figure}

\begin{figure} 
	{\caption{Plot of the non-monotonous function for the event $\{ a \cc b \}$ and $l_x$ and $(n,m)=(18,2)$ (left), $P_x$ and $(n,m) = (2000,300)$ (middle) as well as  $l_x^{\otimes 2}$ and $(n,m) = (38,2)$ (right).} \label{plot} }
	\begin{center}
	{\includegraphics[width= 4 cm]{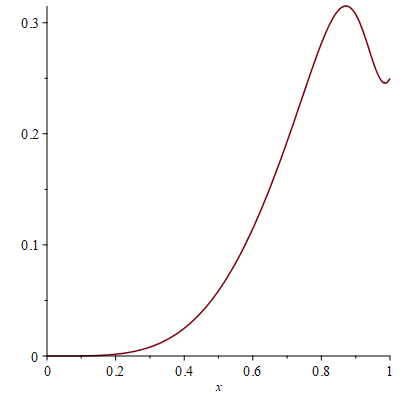}
	\includegraphics[width= 4 cm]{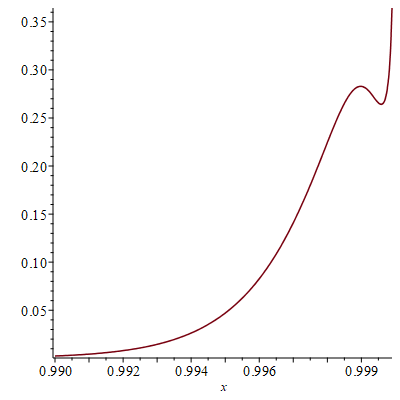}
	\includegraphics[width= 4 cm]{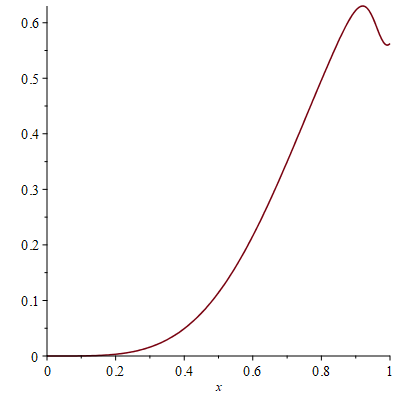}}
	\end{center}
\end{figure}

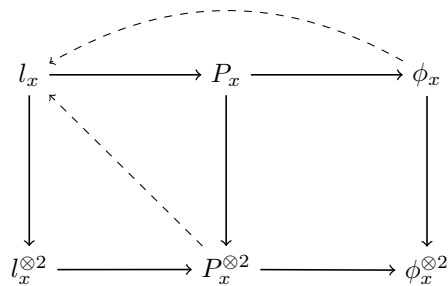
\begin{figure}
		{\caption{Overview of the couplings and correspondingly the implication diagram for monotonicity and FKG in view of Theorem \ref{sumthm}. 
		 Each of the thick lines is either a union of the measure with Bernoulli-percolation (horizontal) or with an independent copy of itself (vertical). Since we have monotonicity and FKG for Bernoulli percolation it follows from Theorem \ref{sumthm} that we obtain the properties also along the thick lines (for some suitable graphs and parameters where we have monotonicity or FKG already).  
		  We have indicated the couplings from Theorem \ref{newcoupling} and  \cite[Corollary 3.6]{GJ09} obtained by picking an even subgraph uniformly at random with dashed lines.  \label{monperco}}
 \begin{center}
 	
\begin{tikzpicture}[font=\sffamily]
\node (X) at (0,0) {
	$l_x$};
\node (Y) [right=2cm of X]  {$P_x$};
\node (Z) [right=2cm of Y]  {$\phi_x$};
\node (A) [below=2cm of X]  {$l_x^{\otimes 2}$
};
\node (B) [below=2cm of Y]  {$P_x^{\otimes 2}$
};
\node (C) [below=2cm of Z]  {$\phi_x^{\otimes 2}$};
\draw [semithick,->] (X) -- (Y);
\draw [semithick,->] (Y) -- (Z);
\draw [semithick,->] (X) -- (A);
\draw [semithick,->] (Y) -- (B);
\draw [semithick,->] (Z) -- (C);
\draw [semithick,->] (A) -- (B);
\draw [semithick,->] (B) -- (C);
\draw [dashed, ->] (B) -- (X);
 \path[dashed,->]
 (Z) edge[bend right] node [left] {} (X);
\end{tikzpicture}
\end{center}}
\end{figure}

\section{Monotonicity and FKG are stable under unions} 
One may observe that the counterexample in Figure \ref{plot} grow larger in size when going from the $l_x$ to $P_x$. The intuition being that the additional independent Bernoulli percolation, relating these models, `enhances' properties of monotonicity and positive association. This intuition is in line with the following theorem which we will use to fill out Table \ref{table:overview}. Further, since we know MON and FKG for percolation the implications along the arrows in Figure \ref{monperco} follow.

\begin{theorem}\label{sumthm}
	Suppose that $\mu_x$ and $\nu_x$ are percolation measures monotonic in $x$. Let
	$\mu_x \cup \nu_x$ have the law of the union of two sets independently sampled. If  $\mu_x$ and $\nu_x$ both satisfy FKG then $\mu_x \cup \nu_x$ satisfies FKG. 
\end{theorem}
\begin{proof}
	If $x_1 < x_2$, then $\mu_{x_2}$ stochastically dominates $\mu_{x_1}$.  We then use Strassen's characterization of stochastic domination.  Let $P_\mu(\eta, \omega) $ be a coupling of the two. Where $\eta \sim \mu_{x_1}$, $\omega \sim \mu_{x_2} $ and $P_\mu( \{ \eta \leq \omega \}) = 1$. Similarly, let $P_\nu( \tilde \eta, \tilde \omega) $ be a coupling of $\nu_{x_1}$ and $\nu_{x_2}$. Then 
	$\mu_{x_1} \cup \nu_{x_1}$ has the law of $\eta \cup \tilde \eta$ whereas  $\mu_{x_2} \cup \nu_{x_2}$ has the law of  $\omega \cup \tilde \omega$. So the measure distributed as $ ( \eta \cup \tilde \eta, \omega \cup \tilde \omega) $ is a coupling with the correct marginals and such that $ \eta \cup \tilde \eta \leq \omega \cup \tilde \omega $ almost surely. This proves the correct stochastic domination. 
	The proof of the FKG-part is given in \cite{Pin19}. 
\end{proof}

\section{A new double current coupling}
The counterexamples presented above do not seem to work for the double current. This  vaguely suggest that the double current is similar in qualitative behaviour to the random cluster model, whereas the single current is more reminiscent of the loop-$\mathrm{O}(1)$-model. The double current also has the same phase transition as the random cluster model, but it is not known to be the case for the single current \cite[Question 1]{DC16}.  This intuition is further supported by the following new coupling for the double current which is exactly the same coupling as for the random cluster model  \cite[Corollary 3.6]{GJ09}. 

\begin{theorem}\label{newcoupling}
	Sample a uniform even subgraph from the traced double current. That has the law of the loop-$\mathrm{O}(1)$-model. 
\end{theorem}
\begin{proof}
	We use the characterization from Theorem 3.2 in \cite{Lis} which gives that
	\begin{align*}
	P_x^{\otimes 2}(\omega) = \frac{1}{Z^2} \abs{ \mathcal{E}_{\emptyset}(\omega)} \sum_{\omega_1 \subset \omega,  \omega_1 \in \mathcal{E}_\emptyset(G)}  x^{\abs{\omega_1}}x^{2 \abs{\omega_2}} (1-x^2)^{\abs{E} - \abs{\omega}}.
	\end{align*}
	Here $\omega_2 = \omega \slash \omega_1$ and compared to Theorem 3.2 in \cite{Lis} there is an additional sum since a slightly different measure is considered there.
The probability of sampling an even subgraph $\eta$ is given as
	\begin{align*}
	P_{coup}(\eta) 
	& =  \sum_{\omega \supset \eta } \frac{1}{ \abs{\mathcal{E}_{\emptyset}(\omega)}} P_x^{\otimes 2}(\omega)  \\
	&  = \frac{1}{Z^2} \ \sum_{\omega} \sum_{\omega_1 \in \mathcal{E}_\emptyset(G)}  1_{\eta \subset \omega} 1_{\omega_1 \subset \omega}x^{\abs{\omega_1}}x^{2 \abs{\omega_2}} (1-x^2)^{\abs{E} - \abs{\omega}}.
	\end{align*}
	Now, we can interchange the two sums and then sum over all edges from $\omega_1 \cup \eta$ and upwards.
	To do that, we let $k$ denote the number of open edges in $\omega$ addition to $\omega_1 \cup \eta$. Then 
	\begin{align*}
	P_{coup}(\eta)  &= \frac{1}{Z^2} \ \sum_{\omega_1 \in \mathcal{E}_\emptyset(G)}  \sum_{k=0}^{\abs{E} - \abs{\eta \cup \omega_1}} \binom{\abs{E} - \abs{\eta \cup \omega_1}}{k} x^{\abs{\omega_1}}x^{2(\abs{\eta \cup \omega_1} +k - \abs{\omega_1})} (1-x^2)^{\abs{E} - \abs{\eta \cup \omega_1} -k } \\
	 &= \frac{1}{Z^2} \ \sum_{\omega_1 \in \mathcal{E}_\emptyset(G)} x^{- \abs{\omega_1}} x^{2\abs{\eta \cup \omega_1}} 
	 =\frac{1}{Z^2} \ x^{\abs{\eta}}  \sum_{\omega_1 \in \mathcal{E}_\emptyset(G)}  x^{\abs{\eta \triangle \omega_1}} =\frac{1}{Z^2}  Z x^{\abs{\eta}}  = l_x(\eta)
	\end{align*}
	where we used the bionomial theorem and the fact that for a fixed even subgraph $\eta$ the map $\omega_1 \mapsto \eta \triangle \omega_1$ is a bijection on the set of even subgraphs.
	
\end{proof}
\noindent \textbf{4.1 Monotonicity of 1-edge events.}
We prove that the event $\{ e \text{ is open} \}$ is monotonic in the loop-$\mathrm{O}(1)$-model. We say that $e$ is \emph{cyclic} if $e$ is part of a loop. Hence by an analogous argument as in  \cite[Corollary 3.6]{GJ09} we obtain the following Corollary.

\begin{corollary} \label{cor1}
	It holds that
	\begin{align*}
	 \frac{1}{2} P_x^{\otimes 2} ( e \textnormal{ open, cyclic} )  = l_x( e \textnormal{ open} ) = \frac{1}{2} \phi_{x} ( e \textnormal{ open, cyclic} ) 
	\end{align*}
	and thus by monotonicity of $ \phi_{x} $ one edge events are monotone. 
\end{corollary}
\begin{proof}
For any configuration $\omega$ if $e$ is open and cyclic in $\omega$ it means that it is part of a loop $K$. The map $\eta \mapsto \eta \triangle K$ is a bijection between the set of even subgraphs of $\omega$ that contain $e$ and the even subgraphs of $\omega$  that do not contain $e$. Thus, $e$ is part of exactly half of the even subgraphs of $\omega$ and the probability that $e$ is still open when we pick an even subgraph uniformly at random is exactly $\frac{1}{2}$. The first equality then follows from Theorem \ref{newcoupling} and the second from the similar coupling from the random cluster model in \cite[Theorem 3.1]{GJ09} as in the proof of \cite[Corollary 3.6]{GJ09}. 

\end{proof}

\noindent We now show how monotonicity of the 1-edge event for the other models follows. 
\begin{corollary}
	Suppose that 
	$ p\colon \lbrack 0,1 \rbrack \to \lbrack 0,1 \rbrack, x \mapsto p(x) $ is non-decreasing and differentiable. Then 
\begin{align*} 
	x \mapsto (l_x \cup \mathbb{P}_{p(x)}) (  e \textnormal{ open} )
\text{           as well as             }
x \mapsto (l_x \cup l_x \cup \mathbb{P}_{p(x)}) (  e \textnormal{ open}  )
\end{align*} 
are increasing. 

\end{corollary}
\begin{proof}
	We have that 
	\begin{align*}
	\left(l_x \cup \mathbb{P}_{p(x)}\right)\left( e \textnormal{ open} \right) = \frac{1}{\sum_{g \in \mathcal{E}_\emptyset(G)} x^{\abs{g}}}  \sum_{g \in \mathcal{E}_\emptyset(G)} x^{\abs{g}} \left(1_{e \in g} + p(x) (1_{e  \not \in g}) \right) 
	= l_x( e \textnormal{ open} ) + p(x) (1- l_x( e \textnormal{ open} ) ). 
	\end{align*} 
	Hence it follows that $(l_x \cup \mathbb{P}_{p(x)})( e \textnormal{ open} )$ has positive derivative by Corollary \ref{cor1} since 
	\begin{align*}
	l_x( e \textnormal{ open} )' (1-p(x))  + p(x)'(1-l_x( e \textnormal{ open} )) \geq 0. 
	\end{align*}
	To see the second part notice that $l_x^{\otimes 2} ( e \textnormal{ open} ) = l_x( e \textnormal{ open} ) \left( 2- l_x( e \textnormal{ open} ) \right) $ and hence $l_x^{\otimes 2} ( e \textnormal{ open} )' \geq 0$.
\end{proof}

\section{Discussion}
On trees  $P_x^{\otimes 2} \sim \mathbb{P}_{x^2}$ and $\phi_x \sim \mathbb{P}_x$ so the models are not the same, but we might ask if the law of the cyclic edges is the same for the two models now that the probabilities that an edge is part of a loop is the same. To see that this is not the case consider again the example in Figure \ref{fkgcounter}.
First, we split up a configuration $\omega = \omega_c \dot \cup \omega_s$ where $\omega_c$ are all the cyclic edges in $\omega$ and $\omega_s$ are the rest. 
Then look at 
\begin{align*}
P_x^{\otimes 2} ( \abs{\omega_c} = l+m ) = \frac{1}{Z^2} \left( x^{2(l+m)} + 2x^{l+m} + x^{2(l+m)}  \right) \cdot (1-x^{2n})
\end{align*}
whereas
\begin{align*}
\phi_x( \abs{\omega_c} = l+m ) = \frac{1}{Z} (x^{l+m}+ x^{l+m})(1-x^n)
\end{align*}
and hence
\begin{align*}
\frac{\phi_x( \abs{\omega_c} = l+m )}{P_x^{\otimes 2} ( \abs{\omega_c} = l+m )}  = 
 \underset{Z}{\underbrace{\left( 1+x^{n+l} + x^{n+m} + x^{l+m }\right)}}  \frac{2 x^{l+m}}{(1+x^n) 2 x^{l+m} (1+x^{l+m})} \neq 1.
\end{align*}
We conclude that the distribution of the loops are not quite the same for the double current and the random cluster model.

With the results at hand, we can now summarize our findings.
In conclusion from the basic results in the introduction, the counterexamples and Theorem \ref{sumthm} used as shown on Figure \ref{monperco} we can establish the properties in Table \ref{table:nonlin}.
 With this overview at hand, it is natural to ask add the following questions in addition to  Question \ref{dcq}. 

\begin{question} \label{quest}
	Is the double random current $P_x^{\otimes 2}$ monotonic, does it satisfy FKG? 
\end{question}

\noindent Notice that our counterexamples do not work for $P_x^{\otimes 2}$. 
Further, monotonicity of the random cluster model can be proven using Holley's criterion for example in the form of Lemma 1.5 in \cite{DC17} . However, this fails for the double random current model for example on the graph considered in Figure \ref{fkgcounter}. 
It is further noted that the FKG-lattice condition is not satisfied \cite{Lis19}.

Let us mention that apart from the coupling shown above and the fact that we know that SING is satisfied then it also holds that $x^{o(g_1) + o(g_2)} p(x)^{\abs{E} - o(g_1 \cup g_2)}  = x^{2 \abs{E}} x^{-o(g_1 \triangle g_2) } $ whenever $p(x) = x^2$. Hence the double random current plays very well together with the structure on the space of even subgraphs. 
From the coupling we see that the loops of the double current are related to the loops of the random cluster model and we hope that the couplings described here could help resolving Question \ref{quest}.  Further, it might inspire couplings also for the double-current measure with sources as the one in \cite[Thm 3.2]{ADTW19}.

\section*{Acknowledgments}
The author would like to thank the Swiss European Mobility Exchange program as well as the Villum Foundation for support through the QMATH center of Excellence(Grant No. 10059) and the Villum Young Investigator (Grant No. 25452) programs.  
Further, thanks to Peter Wildemann for discussions and proofreading and to the anonymous referee for many helpful comments. 

\bibliographystyle{unsrt}

\bibliography{mono1bib2}


\section*{Appendix}

\subsection*{A1: Counterexample to FKG for $l_x^{\otimes 2}$}

\textit{Double loop-$\mathrm{O}(1)$-model:}
We use the same sets as in the counterexample for $l_x$. In Tables \ref{table:lX1} and \ref{table:lX2} all the possible configurations and whether they belong to $X_1$ and $X_2$ are listed. Thus, the corresponding probabilties are 
\begin{align*}
l_x^{\otimes 2}(X_1) = \frac{2x^{n+m} + 3 x^{2(n+m)} + 4 x^{3n+m}}{Z^2} = l_x^{\otimes 2}(X_2)
\end{align*}
whereas 
\begin{align*}
l_x^{\otimes 2}(X_1 \cap X_2) = \frac{2x^{2(n+m)} + 4 x^{3n+m}}{Z^2}.
\end{align*}
Hence when we compare and use that $Z = 1 + O(x)$ we see that 
\begin{align*}
l_x^{\otimes 2}(X_1) l_x^{\otimes 2}(X_2) - l_x^{\otimes 2}(X_1 \cap X_2) Z^2 =  4 x^{2n+2m} + o(x^{2n+2m}) - 2 x^{2(n+m)} + o(x^{2n+2m}) = 2 x^{2n+2m } + o(x^{2n+2m})  
\end{align*}
which is positive for $x$ sufficiently small. Again since $Z \geq 1$ this means that FKG is not satisified.

\subsection*{A2: Explicit formulas for $l_x^{\otimes 2}(a \cc b)$ and $P_x(a \cc b)$}

\textit{Single random currents:} To see that for single random currents SING is not generally satisfied we continue the example from the loop-$\mathrm{O}(1)$-model which is shown in Figure \ref{counter} with $n=l$.  and in addition sample percolation with probability  $p(x) = 1- \sqrt{1-x^2}$ on each edge. Let also $m' = \frac{m}{2}$ then we can calculate the probabilities using the coupling.

The main difficulty is when we get the empty subgraph. In that case we split up and count after how often we open all the edges in the segments of length $m'$ (from $a,b$ to the one of the vertices with valence four). 
If we open three of more $m'$-segments $a$ and $b$ will be connected - this happens with probability $4 p^{3m'}(1-p^{m'}) + p^{4m'}$. If we open at most one $m'$-segment $a$ and $b$ are not connected. Further, there are 4 ways of opening exactly two $m'$ segments each with probability $p^{2m'}(1-p^{m'})^2$ in two of the combinations $a$ and $b$ are connected and in the other two they are connected with probability $2p^n - p^{2n}$. So if we in the loop-$\mathrm{O}(1)$-model sample the empty subgraph the probability that $a$ and $b$ are connected after sampling percolation is $f(p) = \left(4 p^{3m'}(1-p^{m'}) + p^{4m'} +  p^{2m'}(1-p^{m'})^2 \left(2 + 2(2p^n - p^{2n}) \right) \right)$.

That means the total probability becomes
\begin{align*}
P_x(a \cc b) =  \frac{1 \cdot f(p)  + x^{2n} \cdot (2p^{m'} - p^{2m'})^2 +x^{2m}+ 4 \cdot x^{n+m}\cdot (2p^{m'} - p^{2m'}) + x^{2n+2m}}{Z}.
\end{align*}

Here we have plotted $P_x(a \cc b)$ in Figure \ref{plot} for $(n,m) = (2000,300)$ where we see that the function is not monotone.

\textit{Double loop-$\mathrm{O}(1)$-model:} For $l_x^{\otimes 2}$ we have to written out the $8 \times 8$ table in Table \ref{table:lX3} where we plot all pairs of even subgraphs. Using the Table we obtain the following function.  
\begin{align*}
l_x^{\otimes 2}( \{ a \cc b \}) = \frac{1}{Z^2} \left( 2 x^{2m} Z - x^{4m} + 2x^{2n+2m} Z - x^{4n+4m} - 2\cdot x^{2n+4m} + 8 x^{2n+2m}   \right).
\end{align*}
We have plotted the function for $(n,m)=(2,18)$ in Figure \ref{plot} showing that monotonicity and in particular SING also fails for $l_x^{\otimes 2}$.

\begin{table}
	\caption{Double Loop-$\mathrm{O}(1)$-model: Overview of when the event $ X_1$ is satisfied. The rows and columns represent the even subgraph sampled in each of the two independent copies of $l_x$. }
	\label{overview}
	\centering
	\begin{tabular}{c | c | c | c | c }
		&  $\emptyset$ & n+m upper & n+m lower & 2n \\
		\hline $\emptyset$ &  &\checkmark & &\\
		\hline  n+m upper &\checkmark &\checkmark &\checkmark & \checkmark\\
		\hline n+m lower & &\checkmark & &\checkmark \\
		\hline  2n &  & \checkmark&  \checkmark& 
	\end{tabular}
	\label{table:lX1}
\end{table}

\begin{table}
	\caption{Double Loop-$\mathrm{O}(1)$-model: Overview of when $ X_1 \cap X_2$ is satisfied. }
	\label{overview}
	\centering
	\begin{tabular}{c | c | c | c | c }
		&  $\emptyset$ & n+m upper & n+m lower & 2n \\
		\hline $\emptyset$ &  & & &\\
		\hline  n+m upper && &\checkmark & \checkmark\\
		\hline n+m lower & &\checkmark & &\checkmark \\
		\hline  2n &  & \checkmark&  \checkmark& 
	\end{tabular}
	\label{table:lX2}
\end{table}

\begin{table}
	\caption{Double Loop-$\mathrm{O}(1)$-model and all pairs of even subgraphs. We list whether $ \{ a \cc b \}$. With the arrows we indicate if the upper or low path of length $n$ or $m$ is open.}
	\label{dlm}
	\centering
	\begin{tabular}{c | c | c | c | c | c | c | c | c }
		&  $\emptyset$ &  $2m$ & $n \uparrow + m \uparrow$ & $ n \uparrow + m \downarrow$ & $ n \downarrow + m \uparrow$ & $n \downarrow + m \downarrow $& $2n$ & $2n + 2m$ \\
		\hline $\emptyset$ &  & \checkmark & & & &  & & \checkmark \\
		\hline $2m$ & \checkmark  & \checkmark& \checkmark & \checkmark&\checkmark &\checkmark & \checkmark& \checkmark\\
		\hline $n \uparrow + m \uparrow$  &  &\checkmark & & \checkmark & &\checkmark & &\checkmark \\
		\hline $ n \uparrow + m \downarrow$ &  & \checkmark&\checkmark & &\checkmark & & & \checkmark\\	
		\hline $ n \downarrow + m \uparrow$ &  & \checkmark& & \checkmark& &\checkmark & & \checkmark\\
		\hline $n \downarrow + m \downarrow $ &  & \checkmark& \checkmark& &\checkmark & & & \checkmark \\	
		\hline $2n$ &  &\checkmark & & & & & & \checkmark \\
		\hline $2n +2m$ & \checkmark &\checkmark & \checkmark &\checkmark &\checkmark & \checkmark& \checkmark& \checkmark\\
	\end{tabular}
	\label{table:lX3}	
	
\end{table}

\subsection*{A3: Couplings}
In this appendix, we show how the definitions above correspond to the usual definitions of the (traced, sourceless) single random current and the random cluster model. Thereby giving a slightly different and self-contained proof of the relations from \cite[Exercise 36]{DC17}  see also \cite[Theorem 3.5]{GJ09}, \cite{LW16} \cite[Theorem 3.1]{Lis}. Using $x = \tanh(\beta)$ recovers our definitions. 

\begin{theorem}
Defining the (traced, sourceless) single random current and the random cluster model as in \cite{DC17} it holds that
\begin{itemize}
\item $l_{\tanh(\beta)} \cup \mathbb{P}_{ 1 - \cosh(\beta)^{-1}} = P_{\tanh(\beta)} $ 
\item $l_{\tanh(\beta)} \cup \mathbb{P}_{\tanh(\beta)} = \phi_\beta $. 
\end{itemize} 
\end{theorem} 
\begin{proof}
 To prove the first relation, consider all sourceless currents $m$ such that $\hat{m} = n$.
For each sourceless current $m$ let  $u(m)$ be the set of edges with an uneven number. Since the current is sourceless $u(m)$ has to be an even spanning subgraph and conversely for every even spanning subgraph $\gamma$ gives rise to some $m$ with $u(m) = \gamma$. Summing the weight of all such $m$ we find that the relative weight of $\gamma$ is
\begin{align*}
\sum_{m \vert u(m) = \gamma} w(m)  = 
\left( \prod_{e \in \gamma} \sum_{n_e \geq 0} \frac{\beta^{2n_e+1}}{(2n_e +1)!} \right)  \left(  \prod_{e \not \in \gamma}
 \sum_{n_e \geq 0} \frac{\beta^{2n_e}}{(2n_e)!} \right) = \sinh(\beta)^{o(\gamma)} \cosh(\beta)^{\abs{E} - o(\gamma)} \propto \tanh(\beta)^{o(\gamma)}.
\end{align*}
Notice that this is exactly the relative weight of $\gamma$ in the loop-$O(1)$-model. 
Now, given that the uneven edges are $\gamma$ to compute the relative weight of $n$ we need a positive number for all the remaining (even) edges of $n$ and the current has to be 0 at the edges that are closed in $n$. This has weights $\frac{\cosh(\beta) -1}{\cosh(\beta)}  $ and $\frac{1}{\cosh(\beta)}$ for each edge respectively. Since the edges are independent given their parity this is exactly Bernoulli percolation with parameter $1 - \cosh(\beta)^{-1}$.
\\
\\
\\
 For the second relation, notice that $ \frac{x}{1-x} = \frac{\tanh(\beta)}{1- \tanh(\beta)} = \frac{p}{2(1-p)}$ with $p = 1- \exp(- 2 \beta)$. 
Let $\eta$ be an even spanning subgraph of $\omega$. There are $E - o(\eta)$ closed edges left and we have to open $o(\omega) - o(\eta)$ of them. Hence $E - o(\omega)$ have to stay closed. Thus,
\begin{align*}
& \sum_{\partial \eta = \emptyset} \mathbb{P}(\omega \mid \eta) l_\beta(\eta)  =  \sum_{\partial \eta = \emptyset, \eta \leq \omega} x^{o(\omega) - o(\eta)} (1-x)^{E- o(\omega)}  x^{o(\eta)} \\
& \propto \left(\frac{x}{1-x}\right)^{o(\omega)} \sum_{\partial \eta = \emptyset, \eta \leq \omega}1  \propto \left(\frac{p}{2(1-p)}\right)^{o(\omega)} 2^{k(\omega) + o(\omega)}
 \propto p^{o(\omega)} (1-p)^{c(\omega)} 2^{k(\omega)} \propto \phi_p(\omega)
\end{align*}
where we used that a graph with $k$ connected components, $\abs{E}$ edges and $\abs{V}$ vertices has $2^{k + \abs{E} - \abs{V}}$ even spanning subgraphs.

\end{proof}
\noindent It now follows that 
\begin{align*}
P_x^{\otimes 2}  = P_x \cup P_x = l_x \cup \mathbb{P}_{1- \sqrt{1-x^2}} \cup  l_x \cup \mathbb{P}_{1- \sqrt{1-x^2}} = l_x^{\otimes 2} \cup  \mathbb{P}_{2(1- \sqrt{1-x^2}) - (1- \sqrt{1-x^2})^2}  =  l_x^{\otimes 2} \cup \mathbb{P}_{x^2}
\end{align*}
where we used that $\mathbb{P}_p \cup  \mathbb{P}_q =  \mathbb{P}_{p+q - pq} $  as stated two lines below (3). That $\phi_x^{\otimes 2} = l_x^{\otimes 2} \cup \mathbb{P}_{x(2-x)}$ follows by a similar and simpler computation.

\end{document}